\documentclass[runningheads]{llncs}
\usepackage{times}
\usepackage{amsmath,amssymb,txfonts}
\usepackage{xytree}
\usepackage{listings}
\usepackage{algpseudocode}
\usepackage{graphicx}
\usepackage{xspace}
\usepackage{tabularx}
\usepackage{paralist}
\usepackage{multirow}
\usepackage{calc}
\usepackage{bbm}
\usepackage{cite}
\usepackage{tikz}
\usetikzlibrary{arrows,positioning,fit,shapes,calc,trees}

\tikzset{hornNode/.style={draw=black!50,fill=black!5,line width=1pt,rounded corners=2pt},
         hornEdge/.style={line width=2pt,black!50,-latex},
         intHornEq/.style={line width=5pt,black!10}}
\pgfdeclarelayer{background}
\pgfsetlayers{background,main}

\lstdefinelanguage{scala}{
  alsoletter={@,=,>,:,-},
  morekeywords={let,abstract, case, catch, class, def, do, else, extends, final, finally, for, if, implicit, import, match, new, null, object, 
override, package, private, protected, requires, return, sealed, super, this, throw, trait, try, true, type, val, var, while, yield, domain, Boolean,
postcondition, precondition,invariant, constraint, assert, forAll, in, _, return, @generator, ensure, require, ensuring, given, have, =>, continue, returns, havoc,:-},
  sensitive=true,
  morecomment=[l]{//},
  morecomment=[s]{/*}{*/},
  morestring=[b]"
}

\newcommand{\codestyle}{\small\sffamily}

\renewcommand{\emph}[1]{\textit {#1}}

\newcommand{\cfgstate}[1]{*++=<1pc>[o][F-]{#1}}

\newif\ifLongVersion\LongVersionfalse

\newcommand{\constraint}{constraint}
\newcommand{\constraintLang}{\constraint\ language}

\newcommand{\Constraint}{Constraint}
\newcommand{\ConstraintLang}{\Constraint\ language}
\newcommand{\ConstraintLangs}{\ConstraintLang s}
\newcommand{\Con}{\ensuremath{\mathit{Constr}}} 
\newcommand{\constraints}{\constraint s}
\newcommand{\relsym}{relation symbol}
\newcommand{\relsyms}{\relsym s}

\newcommand{\unicl}{\ensuremath{\mathit{Cl}_\forall}}

\newcommand{\fv}[1]{\ensuremath{\mathit{fv}(#1)}}

\newcommand{\expl}{\leftarrow}
\newcommand{\pand}{ \land}

\newcommand{\ClauseSet}{{\cal HC}}

\makeatletter

\DeclareRobustCommand{\myEnsuremath}{%
  \ifmmode
    \expandafter\@firstofone
  \else
    \expandafter\@myEnsuredmath
  \fi}
\long\def\@myEnsuredmath#1{\m{\relax#1}}

\makeatother

\newcommand{\m}[1]{\mbox{$#1$}}   

\algtext*{EndWhile}
\algtext*{EndIf}
\algtext*{EndFor}
\algtext*{EndProcedure}
\algtext*{EndFunction}

\title{The Relationship between Craig Interpolation and Recursion-Free Horn Clauses}

\author{Philipp R\"ummer\inst{1} 
\and Hossein Hojjat\inst{2} \and 
  Viktor Kuncak\inst{2}}
\authorrunning{R\"ummer, Hojjat, Kuncak}
\institute{Uppsala University, Sweden \and
Swiss Federal Institute of Technology Lausanne (EPFL)}

\begin{document}
\maketitle

\begin{abstract}
  Despite decades of research, there are still a number of concepts
  commonly found in software programs that are considered challenging
  for verification: among others, such concepts include concurrency,
  and the compositional analysis of programs with procedures.  As a
  promising direction to overcome such difficulties, recently the use
  of Horn constraints as intermediate representation of software
  programs has been proposed. Horn constraints are related to Craig
  interpolation, which is one of the main techniques used to construct
  and refine abstractions in verification, and to synthesise inductive
  loop invariants. We give a survey of the different forms of Craig
  interpolation found in literature, and show that all of them
  correspond to natural fragments of (recursion-free) Horn
  constraints. We also discuss techniques for solving systems of
  recursion-free Horn constraints.
\end{abstract}

\section{Introduction}

Predicate abstraction \cite{GrafSaidi97ConstructionAbstractStateGraphsPVS} has emerged as a prominent and effective way for model checking software systems. 
A key ingredient in predicate abstraction is analyzing the spurious counter-examples to refine abstractions \cite{BallETAL02RelativeCompletenessAbstractionRefinement}. 
The refinement problem saw a significant progress when Craig interpolants extracted from unsatisfiability proofs were used as relevant predicates \cite{Henzinger:2004}.
While interpolation has enjoyed a significant progress for various logical constraints \cite{DBLP:journals/tocl/CimattiGS10,bonacina12,iprincess2011,tree-interpolants}, 
there have been substantial proposals for more general forms of interpolation \cite{tree-interpolants,DBLP:conf/popl/HeizmannHP10,DBLP:conf/sas/AlbarghouthiGC12}.

As a promising direction to extend the reach of automated verification
methods to programs with procedures, and concurrent programs, among
others, recently the use of Horn constraints as intermediate
representation has been proposed
\cite{DBLP:conf/popl/GuptaPR11,DBLP:conf/pldi/GrebenshchikovLPR12,duality}.
This report examines the relationship between various forms of Craig
interpolation and syntactically defined fragments of recursion-free
Horn clauses. We systematically examine binary interpolation,
inductive interpolant sequences, tree interpolants, restricted DAG
interpolants, and disjunctive interpolants, and show the
recursion-free Horn clause problems to which they correspond.  We
present algorithms for solving each of these classes of problems by
reduction to elementary interpolation problems. We also give a
taxonomy of the various interpolation problems, and the corresponding
systems of Horn clauses, in terms of their computational complexity.



\section{Related Work}

The use of \textbf{Horn clauses} as intermediate representation for
verification was proposed in \cite{lopstr07}.
The authors is \cite{DBLP:conf/popl/GuptaPR11} use Horn clauses 
for verification of multi-threaded programs.
The underlying procedure for solving sets of recursion-free Horn clauses,
over the combined theory of linear integer arithmetic and
uninterpreted functions, was presented in \cite{DBLP:conf/aplas/GuptaPR11}. 
A range of further applications of Horn clauses, including inter-procedural model checking, 
was given in \cite{DBLP:conf/pldi/GrebenshchikovLPR12}. 
Horn clauses are also proposed as intermediate/exchange format for verification problems in
\cite{verificationAsSMT}, and are natively supported by the SMT solver Z3~\cite{Moura:2008}.

There is a long line of research on \textbf{Craig interpolation}
methods, and generalised forms of interpolation, tailored to
verification. For an overview of interpolation in the presence of
theories, we refer the reader to
\cite{DBLP:journals/tocl/CimattiGS10,iprincess2011}. Binary Craig
interpolation for implications~$A \to C$ goes back to
\cite{craig1957linear}, was carried over to conjunctions~$A \wedge B$
in \cite{DBLP:conf/cav/McMillan03}, and generalised to inductive
sequences of interpolants
in~\cite{Henzinger:2004,DBLP:conf/cav/McMillan06}. The concept of tree
interpolation, strictly generalising inductive sequences of
interpolants, is presented in the documentation of the interpolation
engine iZ3~\cite{tree-interpolants}; the computation of tree
interpolants by computing a sequence of binary interpolants is also
described in \cite{DBLP:conf/popl/HeizmannHP10}. Restricted DAG
interpolants~\cite{DBLP:conf/sas/AlbarghouthiGC12} and disjunctive
interpolants~\cite{disjInterpolantsTR} are a further generalisation of
inductive sequences of interpolants, designed to enable the
simultaneous analysis of multiple counterexamples or program
paths.

The use of Craig interpolation for solving Horn clauses is discussed
in \cite{duality}, concentrating on the case of tree interpolation.
Our paper extends this work by giving a systematic study of the
relationship between different forms of Craig interpolation and Horn
clauses, as well as general results about solvability and
computational complexity, independent of any particular calculus used
to perform interpolation. 

\textbf{Inter-procedural software model checking} with interpolants
has been an active area of research for the last decade.  In the
context of predicate abstraction, it has been discussed how
well-scoped invariants can be inferred~\cite{Henzinger:2004} in the
presence of function calls. Based on the concept of Horn clauses, a
predicate abstraction-based algorithm for bottom-up construction of
function summaries was presented in
\cite{DBLP:conf/pldi/GrebenshchikovLPR12}.  Generalisations of the
Impact algorithm~\cite{DBLP:conf/cav/McMillan06} to programs with
procedures are given in \cite{DBLP:conf/popl/HeizmannHP10} (formulated
using nested word automata) and
\cite{DBLP:conf/vmcai/AlbarghouthiGC12}.  Finally, function summaries
generated using interpolants have also been used to speed up bounded
model checking~\cite{sfs2012}.

Several other tools handle procedures by increasingly inlining and performing
under and/or over-approximation \cite{SuterETAL11SatisfiabilityModuloRecursivePrograms,
DBLP:conf/cav/LalQL12,DBLP:journals/ase/TaghdiriJ07}, but without the use 
of interpolation techniques.


\section{Example}
\label{sec:example}

\begin{figure}[tb]
\begin{minipage}{3.3cm}
\begin{footnotesize}
\begin{lstlisting}
def f(n : Int)
     returns rec : Int = {
  if (n > 0) {
    tmp = f(n-1)
    rec = tmp  + 1
  } else { 
    rec = 1
  }
}
def main() {
  var res : Int
  havoc(x: Int $\ge$ 0)
  res = f(x)  
  assert(res == x + 1)
}
\end{lstlisting}
\end{footnotesize}
\end{minipage}
~~
\begin{minipage}{5.2cm}
\xymatrix@1@=20pt{
\xyoption{frame}
&&\ar[d]&&&&\ar[d]\\
&&\cfgstate{q_1} \ar[d]_{\txt{{\small $\mbox{havoc }(x)\wedge x'\ge 0$}}} &&&&
\cfgstate{q_5} \ar@/_1pc/[dll]_{\txt{{\small $n{>}0$}}}
\ar@/_-1pc/[dr]^{\txt{{\small $\neg(n{>}0)$}}}\\
&& \cfgstate{q_2} \ar[d]_{\txt{{\small $res'{=}f(x)$}}}
&&
\cfgstate{q_6} \ar[d]^{\txt{{\small $tmp'{=}f(n-1)$}}}
&&&
\cfgstate{q_8} \ar@/_-1pc/[ddl]^{\txt{{\small $rec'{=}1$}}}\\
&& \cfgstate{q_3} \ar[dr]|{\txt{{\small $res{\neq}x+1$}}}
\ar[dl]_{\txt{{\small $res{=}x+1$}}}
&&
\cfgstate{q_7} \ar@/^-0.5pc/[drr]_{\txt{{\small $rec'{=}tmp + 1$}}}
\\
&\cfgstate{q_4}
&&\cfgstate{e}
&&& \cfgstate{q_9}
}
\end{minipage}
\caption{A recursive program and its control flow graph (see
  Sect.~\ref{sec:example}).\label{fig::example}}

\medskip
\begin{lstlisting}[emph={true,false},emphstyle={\bfseries}]
(1)   r1(X, Res)   $\expl~$ true
(2)   r2(X', Res)  $\expl~$  r1(X, Res) $\pand$ X' $\ge$ 0
(3)   r3(X, Res')  $\expl~$  r2(X, Res) $\pand$ rf(X, Res')
(4)   r4(X, Res)   $\expl~$  r3(X, Res) $\pand$ Res $=$ X + 1
(5)   false        $\expl~$  r3(X, Res) $\pand$ Res $\not=$ X + 1

(6)   r5(N, Rec, Tmp)   $\expl~$ true
(7)   r6(N, Rec, Tmp)   $\expl~$ r5(N, Rec, Tmp) $\pand$ N > 0
(8)   r7(N, Rec, Tmp')  $\expl~$ r6(N, Rec, Tmp) $\pand$ rf(N - 1, Tmp')
(9)   r8(N, Rec, Tmp)   $\expl~$ r5(N, Rec, Tmp) $\pand$ N $\le$ 0
(10)  r9(N, Rec', Tmp)  $\expl~$ r7(N, Rec, Tmp) $\pand$ Rec' = Tmp + 1
(11)  r9(N, Rec', Tmp)  $\expl~$ r8(N, Rec, Tmp) $\pand$ Rec' = 1
(12)  rf(N, Rec)        $\expl~$ r9(N, Rec, Tmp)
\end{lstlisting}
\caption{The encoding of the program in Fig.~\ref{fig::example} into a
  set of recursive Horn clauses.\label{fig::exampleClauses}}
\end{figure}

\begin{figure}[tb]
\begin{align*}
  r_1(x, \mathit{res}) &~~\equiv~~ \mathit{true} \\
  r_2(x, \mathit{res}) &~~\equiv~~ x \geq 0 \\
  r_3(x, \mathit{res}) &~~\equiv~~ \mathit{res} = x + 1 \\
  r_4(x, \mathit{res}) &~~\equiv~~ \mathit{true} \\
  r_5(n, \mathit{rec}, \mathit{tmp}) &~~\equiv~~ \mathit{true} \\
  r_6(n, \mathit{rec}, \mathit{tmp}) &~~\equiv~~ n \geq 1 \\
  r_7(n, \mathit{rec}, \mathit{tmp}) &~~\equiv~~ n = \mathit{tmp} \\
  r_8(n, \mathit{rec}, \mathit{tmp}) &~~\equiv~~ n \leq 0 \\
  r_9(n, \mathit{rec}, \mathit{tmp}) &~~\equiv~~
         \mathit{rec} = n + 1 \vee (n \leq 0 \wedge \mathit{rec} = 1) \\
  r_f(n, \mathit{rec}) &~~\equiv~~
         \mathit{rec} = n + 1 \vee (n \leq 0 \wedge \mathit{rec} = 1)
\end{align*}
  \caption{Syntactic solution of the Horn clauses in Fig.~\ref{fig::exampleClauses}.}
  \label{fig:exampleSolution}
\end{figure}

We start with an example illustrating the use of Horn clauses to
verify a recursive program. Fig.~\ref{fig::example} shows an example
of a recursive program, which is encoded as a
set of (recursive) Horn constraints in Fig.~\ref{fig::exampleClauses}.
The function {\ttfamily f} recursively computes the increment of the argument {\ttfamily n} by $1$.

For translation to Horn clauses we assign an uninterpreted relation symbol {\ttfamily r$i$} to each state $q_i$ of the control flow graph. 
The arguments of the relation symbol {\ttfamily r$i$} act as placeholders of the visible variables in the state $q_i$. 
The relation symbol {\ttfamily rf} corresponds to the summary of the function {\ttfamily f}.
In the relation symbol {\ttfamily rf} we do not include the local variable {\ttfamily tmp} in the arguments since it is invisible from outside the function {\ttfamily f}. 
The first argument of {\ttfamily rf} is the input and the second one is the output.
We do not dedicate any relation symbol to the error state $e$.

The initial states of the functions are not constrained at the beginning; they are just implied by $\mathit{true}$. 
The clause that has $\mathit{false}$ as its head corresponds to the assertion in the program. 
In order to satisfy the assertion with the head $\mathit{false}$, the body of the clause should also be evaluated to $\mathit{false}$. 
We put the condition leading to error in the body of this clause to ensure the error condition is not happening.
The rest of the clauses are one to one translation of the edges in the control flow graph.

For the edges with no function calls we merely relate the variables in the previous state to the variables in the next state using the transfer functions on the edges. 
For example, the clause $(2)$ expresses that $\mathit{res}$ is kept unchanged in the transition from $q_1$ to $q_2$ and the value of {\ttfamily x} is greater than or equal to $0$ in $q_2$. 
For the edges with function call we should also take care of the passing arguments and the return values. 
For example, the clause $(3)$ corresponds to the edge containing a function call from $q_2$ to $q_3$.
This clause sets the value of {\ttfamily res} in the state $q_3$ to the return value of the function {\ttfamily f}.
Note that the only clauses in this example that have more than one relation symbols in the body are the ones related to edges with function calls.

The solution of the obtained system of Horn clauses demonstrates the correctness of the program. 
In a  solution each relation symbol is mapped to an expression over its arguments. 
If we replace the relation symbols in the clauses by the expressions in the solution we should obtain only valid clauses. 
In a system with a genuine path to error we cannot find any solution to the system since we have no way to satisfy the assertion clause. 
Fig.~\ref{fig:exampleSolution} gives one
possible solution of the Horn clauses in terms of concrete formulae,
found by our verification tool
Eldarica.\footnote{\url{http://lara.epfl.ch/w/eldarica}}

This paper discusses techniques to automatically construct solutions
of Horn clauses. Although the Horn clauses encoding programs are
typically recursive, it has been observed that the case of
\emph{recursion-free} Horn clauses is instrumental for constructing
verification procedures operating on Horn
clauses~\cite{DBLP:conf/popl/GuptaPR11,DBLP:conf/pldi/GrebenshchikovLPR12,duality}. Sets
of recursion-free Horn clauses are usually extracted from recursive
clauses by means of finite unwinding; examples are given in
Sect.~\ref{sec:tree-interpolants} and \ref{sec:disjInterpolants}.



\section{Formulae and Horn Clauses\label{horn}}

\paragraph{\ConstraintLangs.}
Throughout this paper, we assume that a first-order vocabulary of
\emph{interpreted symbols} has been fixed, consisting of a set~$\cal
F$ of fixed-arity function symbols, and a set~$\cal P$ of fixed-arity
predicate symbols. Interpretation of $\cal F$ and $\cal P$ is
determined by a class~${\cal S}$ of structures $(U, I)$
consisting of non-empty universe~$U$, and a mapping~$I$ that assigns
to each function in $\cal F$ a set-theoretic function over $U$, and to
each predicate in $\cal P$ a set-theoretic relation over $U$. As a
convention, we assume the presence of an equation symbol~``$=$'' in $\cal
P$, with the usual interpretation.  Given a countably infinite
set~$\cal X$ of variables, a \emph{\constraintLang} is a
set~\Con\ of first-order formulae over $\cal F, P, X$
For example, the language of quantifier-free Presburger arithmetic has
${\cal F} = \{+, -, 0, 1, 2, \ldots\}$ and ${\cal P} = \{=, \leq,
|\}$).

A \constraint\ is called \emph{satisfiable} if it holds for some
structure in $\cal S$ and some assignment of the variables~$\cal X$,
otherwise \emph{unsatisfiable}. We say that a set~$\Gamma \subseteq \Con$ of
\constraints\ \emph{entails} a \constraint~$\phi \in \Con$ if
every structure and variable assignment that satisfies all
\constraints\ in $\Gamma$ also satisfies $\phi$; this is denoted by
$\Gamma \models \phi$.

\fv{\phi} denotes the set of free variables in \constraint~$\phi$.
We write $\phi[x_1, \ldots, x_n]$ to state that a \constraint\
contains (only) the free variables $x_1, \ldots, x_n$, and $\phi[t_1,
\ldots, t_n]$ for the result of substituting the terms~$t_1, \ldots,
t_n$ for $x_1, \ldots, x_n$.  Given a \constraint~$\phi$ containing
the free variables~$x_1, \ldots, x_n$, we write~$\unicl (\phi)$ for
the \emph{universal closure} $\forall x_1, \ldots, x_n. \phi$.

\paragraph{Craig interpolation} is the main technique used to
construct and refine abstractions in software model checking.  A
binary interpolation problem is a conjunction~$A \wedge B$ of
constraints. A \emph{Craig interpolant} is a \constraint~$I$ such that
$A \models I$ and $B \models \neg I$, and such that $\fv{I} \subseteq
\fv{A} \cap \fv{B}$. The existence of an interpolant implies that $A
\wedge B$ is unsatisfiable. We say that a \constraintLang\ has the
\emph{interpolation property} if also the opposite holds: whenever $A
\wedge B$ is unsatisfiable, there is an interpolant~$I$. 

\subsection{Horn Clauses}

To define the concept of Horn clauses, we fix a set~$\cal R$ of
uninterpreted fixed-arity \emph{\relsyms,} disjoint from $\cal P$ and $\cal F$. A \emph{Horn clause} is a
formula~$C \wedge B_1 \wedge \cdots \wedge B_n \to H$
where
\begin{itemize}
\item $C$ is a \constraint\ over $\cal F, P, X$;
\item each $B_i$ is an application~$p(t_1, \ldots, t_k)$ of a \relsym\
  $p \in \cal R$ to first-order terms over $\cal F, X$;
\item $H$ is similarly either an
  application~$p(t_1, \ldots, t_k)$ of $p \in \cal R$ to first-order terms,
  or is the \constraint\ $\mathit{false}$.
\end{itemize}
$H$ is called the \emph{head} of the clause, $C \wedge B_1 \wedge
\cdots \wedge B_n$ the \emph{body.} In case $C = \mathit{true}$, we
usually leave out $C$ and just write $B_1 \wedge \cdots \wedge B_n \to
H$. First-order variables (from $\cal X$) in a clause 
are considered implicitly universally quantified; \relsyms\ represent
set-theoretic relations over the universe~$U$ of a structure~$(U, I)
\in \cal S$. Notions like (un)satisfiability and entailment
generalise straightforwardly to formulae with \relsyms.

A \emph{\relsym\ assignment} is a mapping~$\mathit{sol} : {\cal R} \to
\Con$ that maps each $n$-ary \relsym~$p \in \cal R$ to a
\constraint~$\mathit{sol}(p) = C_p[x_1, \ldots, x_n]$ with $n$ free
variables. The \emph{instantiation}~$\mathit{sol}(h)$ of a Horn
clause~$h$ is defined by:
\begin{align*}
  \mathit{sol}\big(C \wedge p_1(\bar t_1) \wedge \cdots \wedge p_n(\bar t_n)
         \to p(\bar t)\big) &~=~
  C \wedge \mathit{sol}(p_1)[\bar t_1] \wedge \cdots \wedge
    \mathit{sol}(p_n)[\bar t_n] \to \mathit{sol}(p)[\bar t]
    \\
  \mathit{sol}\big(C \wedge p_1(\bar t_1) \wedge \cdots \wedge p_n(\bar t_n)
         \to \mathit{false}\big) &~=~
  C \wedge \mathit{sol}(p_1)[\bar t_1] \wedge \cdots \wedge
    \mathit{sol}(p_n)[\bar t_n] \to \mathit{false}
\end{align*}

\begin{definition}[Solvability]
  Let $\ClauseSet$ be a set of Horn clauses over \relsyms\ $\cal R$.
  \begin{enumerate}
  \item $\ClauseSet$ is called \emph{semantically solvable} if for
    every structure $(U, I) \in \cal S$ there is an interpretation of
    the \relsyms~$\cal R$ as set-theoretic relations over $U$ such
    the universally quantified closure~$\unicl(h)$ of every clause~$h \in \ClauseSet$ 
    holds in $(U,I)$.
  \item A $\ClauseSet$ is called \emph{syntactically solvable} if
    there is a \relsym\ assignment $\mathit{sol}$ such that for every
    structure $(U, I) \in \cal S$ and every clause~$h \in \ClauseSet$
    it is the case that $\unicl(\mathit{sol}(h))$ is satisfied.
  \end{enumerate}
\end{definition}

Note that, in the special case when $\cal S$ contains only
one structure, ${\cal S} = \{ (U,I) \}$, semantic
solvability reduces to the existence of relations
interpreting ${\cal R}$ that extend the structure $(U,I)$ in
such a way to make all clauses true. In
other words, Horn clauses are solvable in a structure if and
only if the extension of the theory of $(U,I)$ by relation symbols ${\cal R}$ 
in the vocabulary and by given Horn clauses as axioms is consistent.

A set~$\ClauseSet$ of Horn clauses induces a \emph{dependence
  relation}~$\to_{\ClauseSet}$ on $\cal R$, defining $p \to_{\ClauseSet} q$ if
there is a Horn clause in $\ClauseSet$ that contains $p$ in its head, and
$q$ in the body. The set~$\ClauseSet$ is called \emph{recursion-free} if
$\to_{\ClauseSet}$ is acyclic, and \emph{recursive} otherwise.  In the
next sections we study the solvability problem for recursion-free Horn
clauses and then show how to use such results in general Horn clause
verification
systems.



\section{Generalised Forms of Craig Interpolation}

\begin{table}[tb]
  \begin{center}
    \begin{tabular}{@{~~}p{0.36\linewidth}@{\qquad}p{0.55\linewidth}@{~~}}
      \textbf{Form of interpolation} & \textbf{Fragment of Horn clauses}
      \\[1ex]\hline
      \raisebox{3ex}{}%
      Binary interpolation \cite{craig1957linear,DBLP:conf/cav/McMillan03}\newline
      $A \wedge B$ & 
      Pair of Horn clauses\newline
      $A \to p(\bar x),~ B \wedge p(\bar x) \to \mathit{false}$
      with $\{\bar x\} = \fv{A} \cap \fv{B}$
      \\[1ex]\hline
      \raisebox{3ex}{}%
      Inductive interpolant seq.~\cite{Henzinger:2004,DBLP:conf/cav/McMillan06}\newline
      $T_1 \wedge T_2  \wedge \cdots \wedge T_n$ &
      Linear tree-like Horn clauses\newline
      $T_1 \to p_1(\bar x_1),~~
      p_1(\bar x_1) \wedge T_2  \to p_2(\bar x_2),~~ \ldots$\newline
      with $\{\bar x_i\} = \fv{T_1, \ldots, T_i} \cap \fv{T_{i+1}, \ldots, T_n}$
      \\[1ex]\hline
      \raisebox{3ex}{}%
      Tree interpolants~\cite{tree-interpolants,DBLP:conf/popl/HeizmannHP10} &
      Tree-like Horn clauses
      \\[1ex]\hline
      \raisebox{3ex}{}%
      Restricted
      DAG interpolants~\cite{DBLP:conf/sas/AlbarghouthiGC12} &
      Linear Horn clauses
      \\[1ex]\hline
      \raisebox{3ex}{}%
      Disjunctive interpolants~\cite{disjInterpolantsTR} &
      Body disjoint Horn clauses
    \end{tabular}
  \end{center}
  \caption{Equivalence of interpolation problems and
    systems of Horn clauses.}
  \label{tab:interpolation-horn}
\end{table}

It has become common to work with generalised forms of Craig
interpolation, such as inductive sequences of interpolants, tree
interpolants, and restricted DAG interpolants. We show that a variety
of such interpolation approaches can be reduced to recursion-free Horn
clauses. Recursion-free Horn clauses thus provide a general framework unifying
and subsuming a number of earlier notions. As a side effect, we can
formulate a general theorem about existence of the individual kinds of
interpolants in Sect.~\ref{sec:complexity}, applicable to any
\constraintLang\ with the (binary) interpolation property.

An overview of the relationship between specific forms of
interpolation and specific fragments of recursions-free Horn clauses
is given in Table~\ref{tab:interpolation-horn}, and will be explained
in more detail in the rest of this
section. Table~\ref{tab:interpolation-horn} refers to the following
fragments of recursion-free Horn clauses:
\begin{definition}[Horn clause fragments]
  We say that a finite, recursion-free set~$\ClauseSet$ of Horn clauses
  \begin{enumerate}
  \item is \emph{linear} if the body of each Horn clause contains at
    most one \relsym,
  \item is \emph{body-disjoint} if for each \relsym~$p$ there is
    at most one clause containing $p$ in its body; furthermore, every
    clause contains $p$ at most once;
  \item is \emph{head-disjoint} if for each \relsym~$p$ there is at
    most one clause containing $p$ in its head;
  \item is \emph{tree-like}~\cite{DBLP:conf/aplas/GuptaPR11} if it
    is body-disjoint and head-disjoint.
  \end{enumerate}
\end{definition}


\begin{theorem}[Interpolation and Horn clauses]
  For each line of Table~\ref{tab:interpolation-horn} it holds that:
  \begin{enumerate}
  \item an interpolation problem of the stated form can be
    polynomially reduced to (syntactically) solving a set of Horn
    clauses, in the stated fragment;
  \item solving a set of Horn clauses (syntactically) in the stated
    fragment can be polynomially reduced to solving a sequence of
    interpolation problems of the stated form.
  \end{enumerate}
\end{theorem}


\subsection{Binary Craig Interpolants \cite{craig1957linear,DBLP:conf/cav/McMillan03}}

The simplest form of Craig interpolation is the derivation of a
\constraint~$I$ such that $A \models I$ and $I \models \lnot B$, and
such that $\fv{I} \subseteq \fv{A} \cap \fv{B}$. Such derivation
is typically constructed by efficiently processing the proof of unsatisfiability
of~$A \wedge B$. To encode a binary interpolation problem
into Horn clauses, we first determine the set~$\bar x = \fv{A}
\cap \fv{B}$ of variables that can possibly occur in the
interpolant. We then pick a \relsym~$p$ of arity~$|\bar x|$, and
define two Horn clauses expressing that $p(\bar x)$ is an interpolant:
\begin{equation*}
  A \to p(\bar x),\qquad B \wedge p(\bar x) \to \mathit{false}
\end{equation*}
It is clear that every syntactic solution for the two Horn clauses
corresponds to an interpolant of $A \wedge B$.


\subsection{Inductive Sequences of
  Interpolants~\cite{Henzinger:2004,DBLP:conf/cav/McMillan06}}
\label{sec:inductive-seq}

Given an unsatisfiable conjunction~$T_1 \wedge \ldots \wedge T_n$ (in practice,
often corresponding to an infeasible path in a program), an \emph{inductive
  sequence of interpolants} is a sequence~$I_0, I_1, \ldots, I_n$ of
formulae such that
\begin{enumerate}
  \item $I_0=\mathit{true}$, $I_n=\mathit{false}$,
  \item for all $i \in \{1, \ldots, n\}$, the entailment~$I_{i-1}\wedge
    T_i \models I_i$ holds, and
  \item for all $i \in \{0, \ldots, n\}$, it is the case that
    $\fv{I_i} \subseteq \fv{T_1, \ldots, T_i} \cap \fv{T_{i+1}, \ldots, T_n}$.
\end{enumerate}
While inductive sequences can be computed by repeated computation of
binary interpolants~\cite{Henzinger:2004}, more efficient solvers have
been developed that derive a whole sequence of interpolants
simultaneously
\cite{DBLP:journals/tocl/CimattiGS10,iprincess2011,tree-interpolants}.

\paragraph{Inductive sequences as linear tree-like Horn clauses.}

An inductive sequence of interpolants can straightforwardly be encoded
as a set of linear Horn clauses, by introducing a fresh \relsym~$p_i$
for each interpolant~$I_i$ to be computed. The arguments of the
\relsyms\ have to be chosen reflecting condition~3 of the definition
of interpolant sequences: for each~$i \in \{0, \ldots, n\}$, we assume
that $\bar x_i = \fv{T_1, \ldots, T_i} \cap \fv{T_{i+1}, \ldots,
  T_n}$ is the vector of variables that can occur in
$I_i$. Conditions~1 and 2 are then represented by the following Horn
clauses:
\begin{equation*}
  p_0(\bar x_0),~~
  p_0(\bar x_0) \wedge T_1 \to p_1(\bar x_1),~~
  p_1(\bar x_1) \wedge T_2  \to p_2(\bar x_2),~~ \ldots,~~
  p_n(\bar x_n) \to \mathit{false}
\end{equation*}

\paragraph{Linear tree-like Horn clauses as inductive sequences.}

Suppose $\ClauseSet$ is a finite, recursion-free, linear, and tree-like
set of Horn clauses.  We can solve the system of Horn clauses by
computing one inductive sequence of interpolants for every connected
component of the $\to_{\ClauseSet}$-graph. First, each clause is
normalised in a manner similar to
\cite{DBLP:conf/pldi/GrebenshchikovLPR12}: for every \relsym\ $p$, we
fix a unique vector of variables~$\bar x_p$, and rewrite $\ClauseSet$ such
that $p$ only occurs in the form~$p(\bar x_p)$; this is possible
since $\ClauseSet$ is recursion-free and body-disjoint. 
We then ensure, through renaming, that every variable~$x$ that is not argument of a
\relsym\ occurs in at most one clause. A connected component then
represents Horn clauses
\begin{equation*}
  C_1 \to p_1(\bar x_1),~~
  C_2 \wedge p_1(\bar x_1) \to p_2(\bar x_2),~~
  C_3 \wedge p_2(\bar x_2) \to p_3(\bar x_3),~~\ldots,~~
  C_n \wedge p_n(\bar x_n) \to \mathit{false}~.
\end{equation*}
(If the first or the last of the clauses is missing, we assume that
its constraint is $\mathit{false}$.)
Any inductive sequence of interpolants for $C_1 \wedge C_2 \wedge
C_3 \wedge \cdots \wedge C_n$ solves the clauses.


\subsection{Tree Interpolants~\cite{tree-interpolants,DBLP:conf/popl/HeizmannHP10}}
\label{sec:tree-interpolants}

Tree interpolants strictly generalise inductive sequences of
interpolants, and are designed with the application of
inter-procedural verification in mind: in this context, the tree
structure of the interpolation problem corresponds to (a part of) the
call graph of a program. Tree interpolation problems correspond to
recursion-free tree-like sets of Horn clauses.

Suppose $(V, E)$ is a finite directed tree, writing $E(v, w)$ to
express that the node~$w$ is a direct child of $v$. Further, suppose
$\phi : V \to \Con$ is a function that labels each node~$v$ of the
tree with a formula~$\phi(v)$. A labelling function~$I : V \to \Con$
is called a \emph{tree interpolant} (for $(V, E)$ and $\phi$) if the
following properties hold:
\begin{enumerate}
\item for the root node~$v_0 \in V$, it is the case that $I(v_0) =
  \mathit{false}$,
\item for any node~$v \in V$, the following entailment holds:
  \begin{equation*}
    \phi(v) \wedge
    \bigwedge_{(v, w) \in E} I(w)
    ~\models~
    I(v)~,
  \end{equation*}
\item for any node~$v \in V$, every non-logical symbol (in our case:
  variable) in $I(v)$ occurs both in some formula~$\phi(w)$ for $w$
  such that $E^*(v, w)$, and in some formula~$\phi(w')$ for some $w'$
  such that $\neg E^*(v, w')$. ($E^*$ is the reflexive transitive
  closure of $E$).
\end{enumerate}

Since the case of tree interpolants is instructive for solving
recursion-free sets of Horn clauses in general, we give a result about
the existence of tree interpolants.  The proof of the lemma computes
tree interpolants by repeated derivation of binary interpolants;
however, as for inductive sequences of interpolants, there are solvers
that can compute all formulae of a tree interpolant
simultaneously~\cite{tree-interpolants,DBLP:conf/popl/GuptaPR11,DBLP:conf/aplas/GuptaPR11}.
\begin{lemma}
  \label{lem:tree-interpolants}
  Suppose the \constraintLang~\Con\ that has the interpolation
  property. Then a tree~$(V, E)$ with labelling function~$\phi : V \to
  \Con$ has a tree interpolant~$I$ if and only if $\bigwedge_{v \in V}
  \phi(v)$ is unsatisfiable.
\end{lemma}


\begin{proof}
  ``$\Rightarrow$'' follows from the observation that every
  interpolant~$I(v)$ is a consequence of the conjunction
  $
    \bigwedge_{(v, w) \in E^+} \phi(w)
  $.
  
  ``$\Leftarrow$'': let $v_1, v_2, \ldots, v_n$ be an inverse
  topological ordering of the nodes in $(V, E)$, i.e., an ordering
  such that $\forall i, j.\; (E(v_i, v_j) \Rightarrow i > j)$. We
  inductively construct a sequence of formulae~$I_1, I_2, \ldots,
  I_n$, such that for every $i \in \{1, \ldots, n\}$ the following
  properties hold:
  \begin{enumerate}
  \item the following conjunction is unsatisfiable:
    \begin{equation}
      \label{eq:conj1}
      \bigwedge\{
      I_k \mid k \leq i,\; \forall j.\; (E(v_j, v_k) \Rightarrow j > i)
      \}
      ~\wedge~
      \Big( \phi(v_{i+1}) \wedge \phi(v_{i+2}) \wedge \cdots
      \wedge \phi(v_n) \Big)
    \end{equation}
  \item the following entailment holds:
    \begin{equation*}
      \phi(v_i) \wedge
      \bigwedge_{(v_i, v_j) \in E} I_j
      ~\models~
      I_i
    \end{equation*}
  \item every non-logical symbol in $I_i$ occurs both in a
    formula~$\phi(w)$ with $E^*(v_i, w)$, and in a
    formula~$\phi(w')$ with $\neg E^*(v_i, w')$.
  \end{enumerate}

  Assume that the formulae~$I_1, I_2, \ldots, I_i$ have been
  constructed, for \m{i \in \{0, \ldots, n-1\}}. We then derive
  the next interpolant~$I_{i+1}$ by solving the binary interpolation
  problem
  \begin{multline}
    \label{eq:conj2}
    \Big(
    \phi(v_{i+1}) \wedge
    \bigwedge_{E(v_{i+1}, v_j)} I_j
    \Big)
    ~~\wedge~~\\
    \Big(
    \bigwedge\{
    I_k \mid k \leq i,\; \forall j.\; (E(v_j, v_k) \Rightarrow j > i + 1)
    \}
    ~\wedge~
    \phi(v_{i+2}) \wedge \cdots
    \wedge \phi(v_n) \Big)
  \end{multline}
  That is, we construct $I_{i+1}$ so that the following entailments
  hold:
  \begin{align*}
    &
    \phi(v_{i+1}) \wedge
    \bigwedge_{E(v_{i+1}, v_j)} I_j
    ~\models~
    I_{i+1},
    \\
    &
    \bigwedge\{
    I_k \mid k \leq i,\; \forall j.\; (E(v_j, v_k) \Rightarrow j > i + 1)
    \}
    ~\wedge~
    \phi(v_{i+2}) \wedge \cdots
    \wedge \phi(v_n)
    ~\models~
    \neg I_{i+1}
  \end{align*}
  Furthermore, $I_{i+1}$ only contains non-logical symbols that are
  common to the left and the right side of the conjunction.

  Note that \eqref{eq:conj2} is equivalent to \eqref{eq:conj1},
  therefore unsatisfiable, and a well-formed interpolation problem.
  It is also easy to see that the properties 1--3 hold for $I_{i+1}$.
  Also, we can easily verify that the labelling function $I : v_i
  \mapsto I_i$ is a solution for the tree interpolation problem
  defined by $(V, E)$ and $\phi$.  \qed
\end{proof}



\paragraph{Tree interpolation as tree-like Horn clauses.}

The encoding of a tree interpolation problem as a tree-like set of
Horn clauses is very similar to the encoding for inductive sequences
of interpolants. We introduce a fresh \relsym~$p_v$ for each node~$v
\in V$ of a tree interpolation problem~$(V, E), \phi$, assuming that
for each $v \in V$ the vector~$\bar x_v = \bigcup_{E^*(v, w)}
\fv{\phi(w)} \cap \bigcup_{\neg E^*(v, w)} \fv{\phi(w)}$ represents
the set of variables that can occur in the interpolant~$I(v)$. The
interpolation problem is then represented by the following clauses:
\begin{align*}
  p_0(\bar x_0) \to \mathit{false},\quad
  \Big\{~ \phi(v) \wedge
    \bigwedge_{(v, w) \in E} p_w(\bar x_w)
    \to
    p_v(\bar x_v) ~\Big\}_{v \in V}
\end{align*}

\paragraph{Tree-like Horn clauses as tree interpolation.}

Suppose $\ClauseSet$ is a finite, recursion-free, and tree-like set of
Horn clauses.  We can solve the system of Horn clauses by computing a
tree interpolant for every connected component of the
$\to_{\ClauseSet}$-graph. As before, we first normalise the Horn
clauses by fixing, for every \relsym~$p$, a unique vector of
variables~$\bar x_p$, and rewriting $\ClauseSet$ such that $p$ only
occurs in the form~$p(\bar x_p)$. We also ensure that every
variable~$x$ that is not argument of a \relsym\ occurs in at most one
clause. The tree interpolation graph~$(V, E)$ is then defined by
choosing the set~$V = {\cal R}\cup \{\mathit{false}\}$ of \relsyms\ as
nodes, and the child relation~$E(p, q)$ to hold whenever $p$ occurs as
head, and $q$ within the body of a clause. The labelling
function~$\phi$ is defined by~$\phi(p) = C$ whenever there is a clause
with head symbol~$p$ and \constraint~$C$, and $\phi(p) =
\mathit{false}$ if $p$ does not occur as head of any clause.

\begin{example}
  \label{ex:treeInt}
  We consider a subset of the Horn clauses given in
  Fig.~\ref{fig::exampleClauses}:
\begin{lstlisting}[emph={true,false},emphstyle={\bfseries}]
  (1)   r1(X, Res)   $\expl~$ true
  (2)   r2(X', Res)  $\expl~$  r1(X, Res) $\pand$ X' $\ge$ 0
  (3)   r3(X, Res')  $\expl~$  r2(X, Res) $\pand$ rf(X, Res')
  (5)   false        $\expl~$  r3(X, Res) $\pand$ Res $\not=$ X + 1
  (6)   r5(N, Rec, Tmp)   $\expl~$ true
  (9)   r8(N, Rec, Tmp)   $\expl~$ r5(N, Rec, Tmp) $\pand$ N $\le$ 0
  (11)  r9(N, Rec', Tmp)  $\expl~$ r8(N, Rec, Tmp) $\pand$ Rec' = 1
  (12)  rf(N, Rec)        $\expl~$ r9(N, Rec, Tmp)
\end{lstlisting}

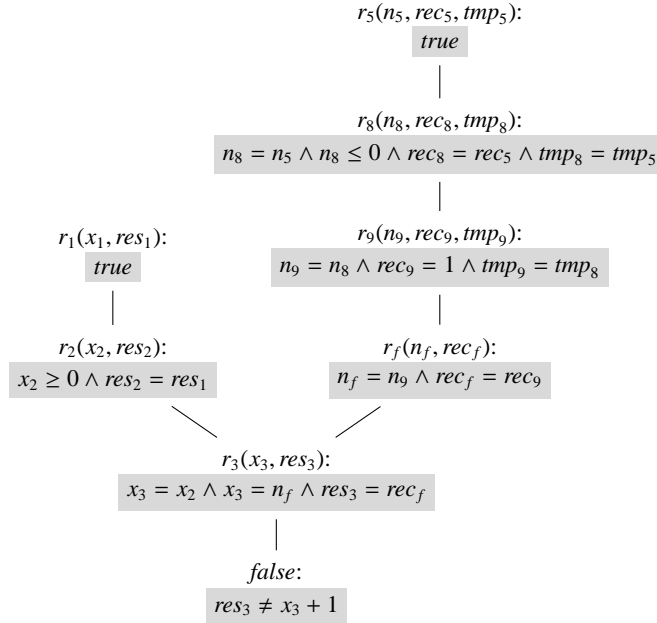
\begin{figure}[tb]
\begin{center}
\begin{tikzpicture}[grow=north]
  \node {\begin{tabular}{c}$\mathit{false}$:\\
           \colorbox[gray]{0.85}{$\mathit{res}_3 \not= x_3 + 1$}\end{tabular}}
      child {node {
          \begin{tabular}{c}
            $r_3(x_3, \mathit{res}_3)$:\\
            \colorbox[gray]{0.85}{$x_3 = x_2 \wedge x_3 = n_f \wedge
                 \mathit{res}_3 = \mathit{rec}_f$}
          \end{tabular}}
        child {node[xshift=10ex] {
          \begin{tabular}{c}$r_f(n_f, \mathit{rec}_f)$:\\
             \colorbox[gray]{0.85}{$n_f = n_9 \wedge
                      \mathit{rec}_f = \mathit{rec}_9$}
          \end{tabular}}
          child {node {
            \begin{tabular}{c}$r_9(n_9, \mathit{rec}_9, \mathit{tmp}_9)$:\\
             \colorbox[gray]{0.85}{$n_9 = n_8 \wedge \mathit{rec}_9 = 1 \wedge
                      \mathit{tmp}_9 =  \mathit{tmp}_8$}
              \end{tabular}}
            child {node {
              \begin{tabular}{c}$r_8(n_8, \mathit{rec}_8, \mathit{tmp}_8)$:\\
                \colorbox[gray]{0.85}{$n_8 = n_5 \wedge n_8 \leq 0 \wedge
                 \mathit{rec}_8 =\mathit{rec}_5 \wedge
                 \mathit{tmp}_8 = \mathit{tmp}_5$}\end{tabular}}
              child {node {
                \begin{tabular}{c}$r_5(n_5, \mathit{rec}_5, \mathit{tmp}_5)$:\\
                 \colorbox[gray]{0.85}{$\mathit{true}$}\end{tabular}}
              }
            }
          }
        }
        child {node[xshift=-10ex] {
            \begin{tabular}{c}$r_2(x_2, \mathit{res}_2)$:\\
             \colorbox[gray]{0.85}{$x_2 \geq 0 \wedge
                   \mathit{res}_2 = \mathit{res}_1$}\end{tabular}}
          child {node {
             \begin{tabular}{c}$r_1(x_1, \mathit{res}_1)$:\\
               \colorbox[gray]{0.85}{$\mathit{true}$}\end{tabular}}
          }
        }
      };
\end{tikzpicture}
\end{center}
  
  \caption{Tree interpolation problem for the clauses in Example~\ref{ex:treeInt}}
  \label{fig:treeInt}
\end{figure}
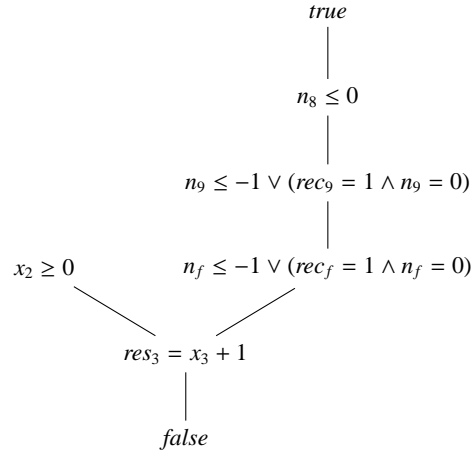
\begin{figure}
\begin{center}
\begin{tikzpicture}[grow=north,level distance=8ex]
  \node {$\mathit{false}$}
     child {node {$\mathit{res}_3 = x_3 + 1$}
       child {node[xshift=8ex] {$n_f \leq -1 \vee (\mathit{rec}_f = 1 \wedge n_f = 0)$}
         child {node {$n_9 \leq -1 \vee (\mathit{rec}_9 = 1 \wedge n_9 = 0)$}
           child {node {$n_8 \leq 0$}
             child {node {$\mathit{true}$}}}
       }}
       child {node[xshift=-8ex] {$x_2 \geq 0$}}
};
\end{tikzpicture}
\end{center}
\caption{Tree interpolant solving the interpolation problem in
  Fig.~\ref{fig:treeInt}}
  \label{fig:treeIntSolution}
\end{figure}

Note that this recursion-free subset of the clauses is body-disjoint
and head-disjoint, and thus tree-like. Since the complete set of
clauses in Fig.~\ref{fig::exampleClauses} is solvable, also any subset
is; in order to compute a (syntactic) solution of the clauses, we set
up the corresponding tree interpolation
problem. Fig.~\ref{fig:treeInt} shows the tree with the
labelling~$\phi$ to be interpolated (in grey), as well as the head
literals of the clauses generating the nodes of the tree.  A tree
interpolant solving the interpolation problem is given in
Fig.~\ref{fig:treeIntSolution}. The tree interpolant can
straightforwardly be mapped to a solution of the original tree-like
Horn, for instance we set $r_8(n_8, \mathit{rec}_8, \mathit{tmp}_8) =
(n_8 \leq 0)$ and $r_9(n_9, \mathit{rec}_9, \mathit{tmp}_9) = (n_9
\leq -1 \vee (\mathit{rec}_9 = 1 \wedge n_9 = 0))$.
\end{example}

\subsubsection{Symmetric Interpolants}

A special case of tree interpolants, \emph{symmetric interpolants,}
was introduced in \cite{DBLP:conf/csl/McMillan04}. Symmetric
interpolants are equivalent to tree interpolants with a flat tree
structure~$(V, E)$, i.e., $V = \{\mathit{root}, v_1, \ldots, v_n\}$,
where the nodes~$v_1, \ldots, v_n$ are the direct children of
$\mathit{root}$.


\subsection{Restricted (and Unrestricted) DAG Interpolants~\cite{DBLP:conf/sas/AlbarghouthiGC12}}

Restricted DAG interpolants are a further generalisation of inductive
sequence of interpolants, introduced for the purpose of reasoning
about multiple paths in a program
simultaneously~\cite{DBLP:conf/sas/AlbarghouthiGC12}.  Suppose $(V, E,
\mathit{en}, \mathit{ex})$ is a finite connected DAG with entry
node~$\mathit{en} \in V$ and exit node~$\mathit{ex} \in V$, further
${\cal L}_E : E \to \Con$ a labelling of edges with constraints, and
${\cal L}_V : V \to \Con$ a labelling of vertices. A \emph{restricted
  DAG interpolant} is a mapping $I : V \to \Con$ with
\begin{enumerate}
\item $I(\mathit{en}) = \mathit{true}$, $I(\mathit{ex}) = \mathit{false}$,
\item for all $(v, w) \in E$ the entailment $I(v) \wedge {\cal L}_V(v)
  \wedge {\cal L}_E(v,w) \models I(w) \wedge {\cal L}_V(w)$ holds, and
\item for all $v \in V$ it is the case that\footnote{The
    definition of DAG interpolants in
    \cite[Def.~4]{DBLP:conf/sas/AlbarghouthiGC12} implies that
    $\fv{I(v)} = \emptyset$ for every interpolant~$I(v), v \in V$,
    i.e., only trivial interpolants are allowed. We assume that this
    is a mistake in \cite[Def.~4]{DBLP:conf/sas/AlbarghouthiGC12}, and
    corrected the definition as shown here.}
  \begin{equation*}
    \fv{I(v)} ~\subseteq~
    \Big(\bigcup_{(a, v) \in E} \fv{{\cal L}_E(a, v)}\Big) \cap \Big(\bigcup_{(v, a)
      \in E} \fv{{\cal L}_E(v, a)}\Big)~.
  \end{equation*}
\end{enumerate}

The UFO verification system~\cite{DBLP:conf/cav/AlbarghouthiLGC12} is
able to compute DAG interpolants, based on the interpolation
functionality of MathSAT~\cite{DBLP:journals/tocl/CimattiGS10}. 
We can observe that DAG interpolants (despite their name) are incomparable in
expressiveness to tree interpolation. 
This is because DAG interpolants correspond to \emph{linear} Horn clauses, and might have shared relation symbol in bodies, while tree interpolants correspond to \emph{possibly nonlinear tree-like} Horn clauses, but do not allow shared relation symbols in bodies.
Nevertheless, it is possible to reduce DAG interpolants to tree interpolants, 
but only at the cost of a potentially exponential growth in the number of clauses. 


\paragraph{Encoding of restricted DAG interpolants as linear Horn
  clauses.}

For every $v \in V$, let
  \begin{equation*}
    \{\bar x_v\} ~~=~~
    \Big(\bigcup_{(a, v) \in E} \fv{{\cal L}_E(a, v)}\Big) \cap
    \Big(\bigcup_{(v, a) \in E} \fv{{\cal L}_E(v, a)}\Big)
  \end{equation*}
  be the variables allowed in the interpolant to be computed for $v$,
  and $p_v$ be a fresh \relsym\ of arity~$|\bar x_v|$. The
  interpolation problem is then defined by the following set of linear
  Horn clauses:
  \begin{align*}
    \text{For each $(v, w) \in E$:}\qquad
    & {\cal L}_V(v) \wedge {\cal L}_E(v, w) \wedge p_v(\bar x_v)
    \to p_w(\bar x_w),
    \\
    & {\cal L}_V(v) \wedge \neg {\cal L}_V(w) \wedge
    {\cal L}_E(v, w) \wedge p_v(\bar x_v)
    \to \mathit{false},
    \\
    \text{For $\mathit{en}, \mathit{ex} \in V$:}\qquad &
    \mathit{true} \to p_{\mathit{en}}(\bar x_{\mathit{en}}),\qquad
    p_{\mathit{ex}}(\bar x_{\mathit{ex}}) \to \mathit{false}
  \end{align*}

\paragraph{Encoding of linear Horn clauses as  DAG interpolants.}
Suppose $\ClauseSet$ is a finite, recursion-free, and linear set of Horn
clauses.  We can solve the system of Horn clauses by computing a DAG
interpolant for every connected component of the $\to_{\ClauseSet}$-graph. 
As in Sect.~\ref{sec:inductive-seq}, we normalise Horn
clauses by fixing a unique vector~$\bar x_p$ of argument variables for
each \relsym~$p$, and ensure that every non-argument variable~$x$
occurs in at most one clause. We also assume that multiple clauses~$C
\wedge p(\bar x_p) \to q(\bar x_q)$ and $D \wedge p(\bar x_p) \to
q(\bar x_q)$ with the same \relsyms\ are merged to $(C \vee D) \wedge
p(\bar x_p) \to q(\bar x_q)$.

Let $\{p_1, \ldots, p_n\}$ be all \relsyms\ of one connected
component. We then define the DAG interpolation problem $(V, E,
\mathit{en}, \mathit{ex}), {\cal L}_E, {\cal L}_V$ by
\begin{itemize}
\item the vertices~$V = \{p_1, \ldots, p_n\} \cup \{\mathit{en},
  \mathit{ex}\}$, including two fresh nodes~$\mathit{en}, \mathit{ex}$,
\item the edge relation
  \begin{align*}
    E ~=~~~~ & \{(p, q) \mid \text{there is a
      clause~} C \wedge p(\bar x_p) \to q(\bar x_q) \in {\ClauseSet} \}
    \\ \cup~ &
     \{(\mathit{en}, p) \mid \text{there is a
       clause~} D \to p(\bar x_p) \in {\ClauseSet} \}
    \\ \cup~ &
     \{(p, \mathit{ex}) \mid \text{there is a
       clause~} E \wedge p(\bar x_p) \to \mathit{false} \in {\ClauseSet} \}~,
  \end{align*}
\item for each $(v, w) \in E$, the edge labelling
  \begin{equation*}
    {\cal L}_E(v, w) ~=~
    \begin{cases}
      C \wedge \bar x_v = \bar x_v \wedge \bar x_w = \bar x_w
      & \text{if~} C \wedge v(\bar x_v) \to w(\bar x_w) \in {\ClauseSet}
      \\
      D \wedge \bar x_w = \bar x_w & \text{if~} v = \mathit{en} \text{~and~}
      D \to w(\bar x_w) \in {\ClauseSet}
      \\
      E \wedge \bar x_v = \bar x_v & \text{if~} w = \mathit{ex} \text{~and~}
      E \wedge v(\bar x_v) \to \mathit{false} \in {\ClauseSet}
    \end{cases}
  \end{equation*}
  Note that the labels include equations like~$\bar x_v = \bar x_v$ to
  ensure that the right variables are allowed to occur in interpolants.
\item for each $v \in V$, the node labelling ${\cal L}_V(v) = \mathit{true}$.
\end{itemize}
By checking the definition of DAG interpolants, it can be verified
that every interpolant solving the problem~$(V, E, \mathit{en},
\mathit{ex}), {\cal L}_E, {\cal L}_V$ is also a solution of the linear
Horn clauses.



\subsection{Disjunctive Interpolants \cite{disjInterpolantsTR}}
\label{sec:disjInterpolants}

Disjunctive interpolants were introduced in \cite{disjInterpolantsTR}
as a generalisation of tree interpolants. Disjunctive interpolants
resemble tree interpolants in the sense that the relationship of the
components of an interpolant is defined by a tree; in contrast to tree
interpolants, however, this tree is an and/or-tree: branching in the
tree can represent either \emph{conjunctions} or \emph{disjunctions.}
Disjunctive interpolants correspond to sets of body-disjoint Horn
clauses; in this representation, and-branching is encoded by clauses
with multiple body literals (like with tree interpolants), while
or-branching is interpreted as multiple clauses sharing the same head
symbol. For a detailed account on disjunctive interpolants, we refer
the reader to \cite{disjInterpolantsTR}.

The solution of body-disjoint Horn clauses can be computed by solving
a sequence of tree-like sets of Horn clauses:
\begin{lemma}
  Let $\ClauseSet$ be a finite set of recursion-free body-disjoint
  Horn clauses. $\ClauseSet$ has a syntactic/semantic solution if and
  only if every maximum tree-like subset of $\ClauseSet$ has a
  syntactic/semantic solution.
\end{lemma}

\begin{proof}
  We outline direction ``$\Leftarrow$'' for syntactic solutions.
  Solving the tree-like subsets of $\ClauseSet$ yields, for each
  \relsym~$p \in \cal R$, a set~$\mathit{SC}_p$ of solution
  constraints. A global solution of $\ClauseSet$ can be constructed by
  forming a positive Boolean combination of the \constraints\ in
  $\mathit{SC}_p$ for each $p \in \cal R$.
 \qed
\end{proof}

\begin{example}
  We consider a recursion-free unwinding of the Horn clauses in
  Fig.~\ref{fig::exampleClauses}.  To make the set of clauses
  body-disjoint, the clause (6), (9), (11), (12) were duplicated,
  introducing primed copies of all \relsyms\ involved. The clauses are
  not head-disjoint, since (10) and (11) share the same head symbol:
\begin{lstlisting}[emph={true,false},emphstyle={\bfseries}]
  (1)   r1(X, Res)   $\expl~$ true
  (2)   r2(X', Res)  $\expl~$  r1(X, Res) $\pand$ X' $\ge$ 0
  (3)   r3(X, Res')  $\expl~$  r2(X, Res) $\pand$ rf(X, Res')
  (5)   false        $\expl~$  r3(X, Res) $\pand$ Res $\not=$ X + 1

  (6)   r5(N, Rec, Tmp)   $\expl~$ true
  (7)   r6(N, Rec, Tmp)   $\expl~$ r5(N, Rec, Tmp) $\pand$ N > 0
  (8)   r7(N, Rec, Tmp')  $\expl~$ r6(N, Rec, Tmp) $\pand$ rf'(N - 1, Tmp')
  (9)   r8(N, Rec, Tmp)   $\expl~$ r5(N, Rec, Tmp) $\pand$ N $\le$ 0
  (10)  r9(N, Rec', Tmp)  $\expl~$ r7(N, Rec, Tmp) $\pand$ Rec' = Tmp + 1
  (11)  r9(N, Rec', Tmp)  $\expl~$ r8(N, Rec, Tmp) $\pand$ Rec' = 1
  (12)  rf(N, Rec)        $\expl~$ r9(N, Rec, Tmp)

  (6')   r5'(N, Rec, Tmp)   $\expl~$ true
  (9')   r8'(N, Rec, Tmp)   $\expl~$ r5'(N, Rec, Tmp) $\pand$ N $\le$ 0
  (11')  r9'(N, Rec', Tmp)  $\expl~$ r8'(N, Rec, Tmp) $\pand$ Rec' = 1
  (12')  rf'(N, Rec)        $\expl~$ r9'(N, Rec, Tmp)
\end{lstlisting}
There are two maximum tree-like subsets: $T_1 = \{(1), (2), (3), (5),
(6), (9), (11), (12)\}$, and $T_2 = \{(1), (2), (3), (5), (6), (7),
(8), (10), (12), (6'), (9'), (11'), (12')\}$. The subset~$T_1$ has
been discussed in Example~\ref{ex:treeInt}. In the same way, it is
possible to construct a solution for $T_2$ by solving a tree
interpolation problem. The two solutions can be combined to construct
a solution of $T_1 \cup T_2$:

\medskip
\noindent
 {\small
\begin{equation*}
  \begin{array}{l@{\qquad}*{3}{@{\quad}l}}
    \hline
    & T_1 & T_2 & T_1 \cup T_2
    \\\hline
    r_1(x, r) &
    \mathit{true} & \mathit{true} & \mathit{true}
    \\
    r_2(x, r) &
    x \geq 0 & \mathit{true} & x \geq 0
    \\
    r_3(x, r) &
    r = x + 1 & r = x + 1 & r = x + 1
    \\
    r_5(n, c, t) &
    \mathit{true} & \mathit{true} & \mathit{true}
    \\\hline
    r_6(n, c, t) &
    - & n \geq 1 & n \geq 1
    \\
    r_7(n, c, t) &
    - & t = n & t = n
    \\
    r_8(n, c, t) &
    n \leq 0 & - & n \leq 0
    \\
    r_9(n, c, t) &
    n \leq -1 \vee (c = 1 \wedge n = 0) & c = n + 1 &
    n \leq -1 \vee c = n + 1
    \\
    r_f(n, c) &
    n \leq -1 \vee (c = 1 \wedge n = 0) & c = n + 1 &
    n \leq -1 \vee c = n + 1
    \\\hline
    r_5'(n, c, t) &
    - & \mathit{true} & \mathit{true}
    \\
    r_8'(n, c, t) &
    - & n \leq 0 & n \leq 0
    \\
    r_9'(n, c, t) &
    - & n \leq -1 \vee (c = 1 \wedge n = 0)
    & n \leq -1 \vee (c = 1 \wedge n = 0)
    \\
    r_f'(n, c, t) &
    - & n \leq -1 \vee (c = 1 \wedge n = 0)
    & n \leq -1 \vee (c = 1 \wedge n = 0)
    \\\hline
  \end{array}
\end{equation*}}

In particular, the disjunction of the two interpretations of $r_9(n,
c, t)$ has to be used, in order to satisfy both (10) and (11)
(similarly for $r_f(n, c)$). In contrast, the conjunction of the
interpretations of $r_2(n, c, t)$ is needed to satisfy (3).
\end{example}



\section{The Complexity of Recursion-free Horn Clauses}
\label{sec:complexity}

\begin{figure}[tb]
  \begin{center}
  \begin{tikzpicture}
    \node[hornNode] (linearTree) {Linear tree-like};
    \node[hornNode] at ($(linearTree)+(0.8,2)$) (bodyD) {Body-disjoint};
    \node[hornNode,above=2.7 of bodyD] (general) {General recursion-free};
    \node[hornNode] at ($(linearTree)+(2.2,1)$) (tree) {Tree-like};
    \node[hornNode,above=2.7 of tree] (headD) {Head-disjoint};
    \node[hornNode] at ($(linearTree)+(-0.8,3)$) (linear) {Linear};

    \draw[hornEdge] (linearTree) -- (tree);
    \draw[hornEdge] (linearTree) -- (linear);
    \draw[hornEdge] (tree) -- (bodyD);
    \draw[hornEdge] (tree) -- (headD);
    \draw[hornEdge] (bodyD) -- (general);
    \draw[hornEdge] (headD) -- (general);
    \draw[hornEdge] (linear) -- (general);

    \node[hornNode,left=1 of linearTree] (intSeq) {Inductive interpolant sequences};
    \node[hornNode] at ($(intSeq)+(0,-1)$) (binInt) {Binary interpolation};
    \node[hornNode] at (intSeq |- tree) (treeInt) {Tree interpolation};
    \node[hornNode] at (intSeq |- bodyD) (disjInt) {Disjunctive interpolation};
    \node[hornNode,xshift=-7ex] at (intSeq |- linear) (dagInt) {(Restricted) DAG interpolation};

    \draw[hornEdge] (binInt) -- (intSeq);
    \draw[hornEdge] (intSeq) -- (treeInt);
    \draw[hornEdge] (treeInt) -- (disjInt);
    \draw[hornEdge] (intSeq) -| ($(dagInt.south)+(-1.5,0)$);

    \begin{pgfonlayer}{background}
      \draw[intHornEq] (intSeq) -- (linearTree);
      \draw[intHornEq] (treeInt) -- (tree);
      \draw[intHornEq] (disjInt) -- (bodyD);
      \draw[intHornEq] (dagInt) -- (linear);

      \draw[line width=2pt,dashed,draw=black!30] ($(linear)+(-6.5,0.5)$) -- +(11.5,0);
      \draw[line width=2pt,dashed,draw=black!30] ($(linearTree)+(-1.7,-2.1)$) -- +(0,7.9);
      \node[rotate=90,anchor=south east] at ($(linear)+(4.9,-1.2)$) {co-NP};
      \node[rotate=90,anchor=south west] at ($(linear)+(4.9,0.7)$) {co-NEXPTIME};
      \node[anchor=north west] at ($(linearTree)+(-1.2,-1.6)$) (recFreeLabel) {Recursion-free Horn clauses};
      \node[anchor=north] at (binInt |- recFreeLabel.north) {Craig interpolation};
    \end{pgfonlayer}
 \end{tikzpicture}
  \end{center}

  \caption{Relationship between different forms of Craig
    interpolation, and different fragments of recursion-free Horn
    clauses. An arrow from A to B expresses that problem A is
    (strictly) subsumed by B. The complexity classes ``co-NP'' and
    ``co-NEXPTIME'' refer to the problem of checking solvability of
    Horn clauses over quantifier-free Presburger arithmetic.}
  \label{fig:complexity}
\end{figure}
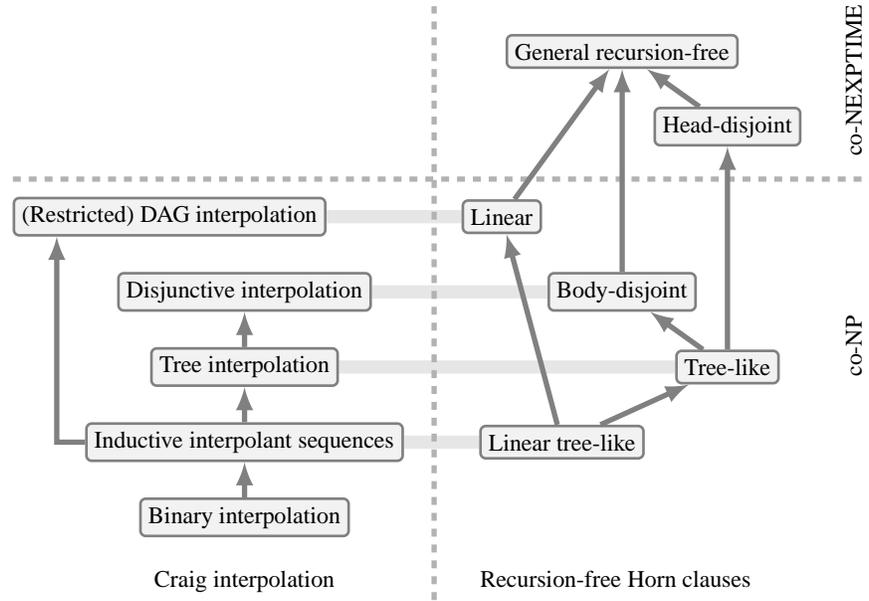

We give an overview of the considered fragments of recursion-free Horn
clauses, and the corresponding interpolation problem, in
Fig.~\ref{fig:complexity}. The diagram also shows the complexity of
deciding (semantic or syntactic) solvability of a set of Horn clauses,
for Horn clauses over the \constraintLang\ of quantifier-free
Presburger arithmetic. Most of the complexity results occur in
\cite{disjInterpolantsTR}, but in addition we use the following two
observations:
\begin{lemma}
  Semantic solvability of recursion-free linear Horn clauses over the
  \constraintLang\ of quantifier-free Presburger arithmetic is in
  co-NP.
\end{lemma}

\begin{proof}
  A set $\ClauseSet$ of recursion-free linear Horn clauses is solvable
  if and only if the expansion~$\mathit{exp}(\ClauseSet)$ is
  unsatisfiable~\cite{disjInterpolantsTR}. For linear clauses,
  $\mathit{exp}(\ClauseSet)$ is a disjunction of (possibly)
  exponentially many formulae, each of which is linear in the size of
  $\mathit{exp}(\ClauseSet)$. Consequently, satisfiability of
  $\mathit{exp}(\ClauseSet)$ is in NP, and unsatisfiability in
  co-NP. \qed
\end{proof}

\begin{lemma}
  Semantic solvability of recursion-free head-disjoint Horn clauses
  over the \constraintLang\ of quantifier-free Presburger arithmetic
  is co-NEXPTIME-hard.
\end{lemma}

\begin{proof}
  The proof given in \cite{disjInterpolantsTR} for
  co-NEXPTIME-hardness of recursion-free Horn clauses over
  quantifier-free Presburger arithmetic can be adapted to only require
  head-disjoint clauses. This is because a single execution step of a
  non-deterministic Turing machine can be expressed as quantifier-free
  Presburger formula. \qed
\end{proof}



\section{Beyond Recursion-free Horn Clauses}
\label{sec:beyond}

It is natural to ask whether the considerations of the last sections
also apply to clauses that are not Horn clauses (i.e., clauses that
can contain multiple positive literals), provided the clauses are
``recursion-free.'' Is it possible, like for Horn clauses, to compute
solutions of recursion-free clauses in general by means of computing
Craig interpolants?

To investigate the situation for clauses that are not Horn, we first
have to generalise the concept of clauses being recursion-free: the
definition provided in Sect.~\ref{horn}, formulated with the help of
the dependence relation~$\to_{\ClauseSet}$, only applies to Horn
clauses. For non-Horn clauses, we instead choose to reason about the
absence of infinite propositional resolution derivations. Because the
proposed algorithms \cite{disjInterpolantsTR} for solving
recursion-free sets of Horn clauses all make use of \emph{exhaustive
  expansion} or \emph{inlining,} i.e., the construction of all
derivations for a given set of clauses, the requirement that no
infinite derivations exist is fundamental.\footnote{We do not take
  subsumption between clauses, or loops in derivations into
  account. This means that a set of clauses might give rise to
  infinite derivations even if the set of derived clauses is
  finite. It is conceivable that notions of subsumption, or more
  generally the application of terminating saturation
  strategies~\cite{FLHT01}, can be used to identify more general
  fragments of clauses for which syntactic solutions can effectively
  be computed. This line of research is future work.}

Somewhat surprisingly, we observe that all sets of clauses without
infinite derivations have the shape of Horn clauses, up to renaming of
\relsyms. This means that procedures handling Horn clauses cover all
situations in which we can hope to compute solutions with the help of
Craig interpolation.

\medskip
Since constraints and \relsym\ arguments are irrelevant for this
observation, the following results are entirely formulated on the
level of propositional logic:
\begin{itemize}
\item a propositional \emph{literal} is either a Boolean variable $p,
  q, r$ (positive literals), or the negation $\neg p, \neg q, \neg r$
  of a Boolean variable (negative literals).
\item a propositional \emph{clause} is a disjunction $p \vee \neg q
  \vee p$ of literals. The multiplicity of a literal is important,
  i.e., clauses could alternatively be represented as multi-sets of
  literals.
\item a \emph{Horn clause} is a clause that contains at most one
  positive literal.
\item given a set~$\ClauseSet$ of Horn clauses, we define the
  dependence relation~$\to_{\ClauseSet}$ on Boolean variables by
  setting \m{p \to_{\ClauseSet} q} if and only if there is a clause in
  $\ClauseSet$ in which $p$ occurs positively, and $q$ negatively
  (like in Sect.~\ref{horn}). The set~$\ClauseSet$ is called
  \emph{recursion-free} if $\to_{\ClauseSet}$ is acyclic.
\end{itemize}

We can now generalise the notion of a set of clauses being
``recursion-free'' to non-Horn clauses:
\begin{definition}
  \label{def:termination}
  A set~$\cal C$ of propositional clauses has the \emph{termination
    property} if no infinite sequence~$c_0, c_1, c_2, c_3, \ldots$
  of clauses exists, such that
  \begin{itemize}
  \item $c_0 \in \cal C$ is an input clause, and
  \item for each $i \geq 1$, the clause~$c_i$ is derived by means of
    binary resolution from $c_{i-1}$ and an input clause, using the
    rule
    \begin{equation*}
      \begin{array}{c}
        C \vee p \qquad D \vee \neg p
        \\\hline
        C \vee D
      \end{array}~.
    \end{equation*}
  \end{itemize}
\end{definition}

\begin{lemma}
  A finite set~$\ClauseSet$ of Horn clauses has the termination
  property if and only if it is recursion-free.
\end{lemma}

\begin{proof}
  ``$\Leftarrow$''
  The acyclic dependence relation~$\to_{\ClauseSet}$ induces a strict
  well-founded order~$<$ on Boolean variables: $q \to_{\ClauseSet} p$
  implies $p < q$. The order~$<$ induces a well-founded order~$\ll$ on
  Horn clauses:
  \begin{align*}
    (p \vee C) \ll (q \vee D) &~~\Leftrightarrow~~
    p > q \text{~~or~~} (p = q \text{~and~} C <_{ms} D)
    \\
    C \ll (q \vee D) &~~\Leftrightarrow~~ \mathit{true}
    \\
    C \ll D &~~\Leftrightarrow~~ C <_{ms} D
  \end{align*}
  where $C, D$ only contain negative literals, and $<_{ms}$ is the
  (well-founded) multi-set extension of $<$
  \cite{DBLP:journals/cacm/DershowitzM79}.

  It is easy to see that a clause~$C \vee D$ derived from two Horn
  clauses~$C \vee p$ and $D \vee \neg p$ using the resolution rule is
  again Horn, and $(C \vee D) \ll (C \vee p)$ and $(C \vee D) \ll (D
  \vee \neg p)$.  The well-foundedness of $\ll$ implies that any
  sequence of clauses as in Def.~\ref{def:termination} is
  finite.

  ``$\Rightarrow$'' If the dependence relation~$\to_{\ClauseSet}$ has
  a cycle, we can directly construct a non-terminating sequence~$c_0,
  c_1, c_2, \ldots$ of clauses. \qed
\end{proof}

\begin{definition}[Renamable-Horn~\cite{Lewis:1978:RSC:322047.322059}]
  If $A$ is a set of Boolean variables, and $\cal C$ is a set of
  clauses, then $r_A({\cal C})$ is the result of replacing in $\cal C$
  every literal whose Boolean variable is in $A$ with its
  complement. $\cal C$ is called \emph{renamable-Horn} if there is
  some set $A$ of Boolean variables such that $r_A({\cal C})$ is Horn.
\end{definition}

\begin{theorem}
  If a finite set $\cal C$ of clauses has the termination property,
  then it is renamable-Horn.
\end{theorem}

\begin{proof}
  Suppose $\cal C$ is formulated over the (finite) set~$p_1, p_2,
  \ldots, p_n$ of Boolean variables.  We construct a graph~$(V, E)$,
  with $V = \{p_1, p_2, \ldots, p_n, \neg p_1, \neg p_2, \ldots, \neg
  p_n\}$ being the set of all possible literals, and $(l, l') \in E$
  if and only if there is a clause~$\neg l \vee l' \vee C \in \cal C$
  (that means, a clause containing the literal~$l'$, and the
  literal~$l$ with reversed sign).\footnote{This graph could
    equivalently be defined as the implication graph of the 2-sat
    problem introduced in \cite{Lewis:1978:RSC:322047.322059}, as a
    way of characterising whether a set of clauses is Horn.}

  The graph~$(V, E)$ is acyclic. To see this, suppose there is a
  cycle~$l_1, l_2, \ldots, l_m, l_{m+1} =l_1$ in $(V, E)$. Then there
  are clauses~$c_1, c_2, \ldots, c_m \in \cal C$ such that each $c_i$
  contains the literals~$\neg l_i$ and $l_{i+1}$. We can then
  construct an infinite sequence~$c_1 = d_0, d_1, d_2, \ldots$ of
  clauses, where each $d_i$ (for $i > 1$) is obtained by resolving
  $d_{i-1}$ with $c_{(i \operatorname{mod} m)+1}$, contradicting the
  assumption that $\cal C$ has the termination property.

  Since $(V, E)$ is acyclic, there is a strict total order~$<$ on $V$
  that is consistent with $E$, i.e., $(l, l') \in E$ implies $l < l'$.

  \emph{Claim:} if $p < \neg p$ for every Boolean variable~$p
  \in \{p_1, p_2, \ldots, p_n\}$, then $\cal C$ is Horn.

  \emph{Proof of the claim:} suppose a non-Horn clause $p_i \vee p_j
  \vee C \in \cal C$ exists (with $i \not= j$). Then $(\neg p_i, p_j)
  \in E$ and $(\neg p_j, p_i) \in E$, and therefore $\neg p_i < p_j$
  and $\neg p_j < p_i$. Then also $\neg p_i < p_i$ or $\neg p_j <
  p_j$, contradicting the assumption that $p < \neg p$ for every
  Boolean variable~$p$.

  In general, choose $A = \{ p_i \mid i \in \{1, \ldots, n\}, \neg p_i
  < p_i \}$, and consider the set~$r_A({\cal C})$ of clauses. The
  set~$r_A({\cal C})$ is Horn, since changing the sign of a Boolean
  variable~\m{p \in A} has the effect of swapping the nodes~$p, \neg
  p$ in the graph~$(V, E)$. Therefore, the new graph~$(V, E')$ has to
  be compatible with a strict total order~$<$ such that $p < \neg p$
  for every Boolean variable~$p$, satisfying the assumption of the
  claim above.
  \qed
\end{proof}

\begin{example}
  We consider the following set of clauses:
  \begin{equation*}
    {\cal C} = \{
      \neg a \vee s,\;  a \vee \neg p,\;  p \vee \neg b,\;  b \vee p \vee r,\;  \neg p \vee q
    \}
  \end{equation*}
  By constructing all possible derivations, it can be shown that the
  set has the termination property.  The graph~$(V, E)$, as
  constructed in the proof, is:
  
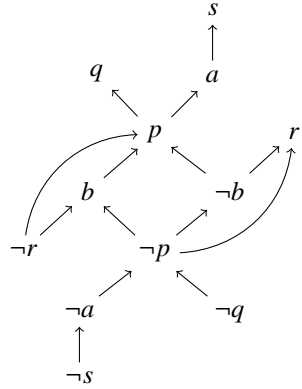
\begin{center}
    \begin{tikzpicture}[node distance=3ex]
      \node (np) {$\neg p$};
      \node[below left=of np] (na) {$\neg a$};
      \node[below=of na] (ns) {$\neg s$};
      \node[below right=of np] (nq) {$\neg q$};
      \node[above left=of np] (b) {$b$};
      \node[above right=of np] (nb) {$\neg b$};
      \node[above=7ex of np] (p) {$p$};
      \node[above right=of p] (a) {$a$};
      \node[above left=of p] (q) {$q$};
      \node[above=of a] (s) {$s$};
      \node[above right=of nb] (r) {$r$};
      \node[below left=of b] (nr) {$\neg r$};

      \draw[->] (ns) -- (na);
      \draw[->] (na) -- (np);
      \draw[->] (nq) -- (np);
      \draw[->] (np) -- (b);
      \draw[->] (np) -- (nb);
      \draw[->] (np) edge[bend right=40] (r);
      \draw[->] (b) -- (p);
      \draw[->] (nb) -- (p);
      \draw[->] (nb) -- (r);
      \draw[->] (p) -- (q);
      \draw[->] (p) -- (a);
      \draw[->] (a) -- (s);
      \draw[->] (nr) -- (b);
      \draw[->] (nr) edge[bend left=40] (p);
    \end{tikzpicture}
  \end{center}
  A strict total order that is compatible with the graph is:
  \begin{equation*}
    \neg s < \neg q < \neg r < \neg a < \neg p < b < \neg b < r < p < q < a < s
  \end{equation*}
  From the order we can read off that we need to rename the variables
  $A = \{s, q, r, a, p\}$ in order to obtain a set of Horn clauses:
  \begin{equation*}
    r_A({\cal C}) = \{
    a \vee \neg s,\;
    \neg a \vee p,\;  \neg p \vee \neg b,\;
    b \vee \neg p \vee \neg r,\;  p \vee \neg q
    \}
  \end{equation*}
\end{example}




{\raggedright


\begin{thebibliography}{10}

\bibitem{DBLP:conf/sas/AlbarghouthiGC12}
A.~Albarghouthi, A.~Gurfinkel, and M.~Chechik.
\newblock Craig interpretation.
\newblock In {\em SAS}, 2012.

\bibitem{DBLP:conf/vmcai/AlbarghouthiGC12}
A.~Albarghouthi, A.~Gurfinkel, and M.~Chechik.
\newblock Whale: An interpolation-based algorithm for inter-procedural
  verification.
\newblock In {\em VMCAI}, pages 39--55, 2012.

\bibitem{DBLP:conf/cav/AlbarghouthiLGC12}
A.~Albarghouthi, Y.~Li, A.~Gurfinkel, and M.~Chechik.
\newblock Ufo: A framework for abstraction- and interpolation-based software
  verification.
\newblock In {\em CAV}, pages 672--678, 2012.

\bibitem{BallETAL02RelativeCompletenessAbstractionRefinement}
T.~Ball, A.~Podelski, and S.~K. Rajamani.
\newblock Relative completeness of abstraction refinement for software model
  checking.
\newblock In {\em TACAS'02}, volume 2280 of {\em LNCS}, page 158, 2002.

\bibitem{verificationAsSMT}
N.~Bj{\o}rner, K.~McMillan, and A.~Rybalchenko.
\newblock Program verification as satisfiability modulo theories.
\newblock In {\em SMT Workshop at IJCAR}, 2012.

\bibitem{bonacina12}
M.~P. Bonacina and M.~Johansson.
\newblock On interpolation in automated theorem proving.
\newblock {\em (submitted)}, 2012.

\bibitem{iprincess2011}
A.~Brillout, D.~Kroening, P.~R{\"u}mmer, and T.~Wahl.
\newblock An interpolating sequent calculus for quantifier-free {Presburger}
  arithmetic.
\newblock {\em Journal of Automated Reasoning}, 47:341--367, 2011.

\bibitem{DBLP:journals/tocl/CimattiGS10}
A.~Cimatti, A.~Griggio, and R.~Sebastiani.
\newblock Efficient generation of {Craig} interpolants in satisfiability modulo
  theories.
\newblock {\em ACM Trans. Comput. Log.}, 12(1):7, 2010.

\bibitem{craig1957linear}
W.~Craig.
\newblock Linear reasoning. {A} new form of the {Herbrand-Gentzen} theorem.
\newblock {\em The Journal of Symbolic Logic}, 22(3):250--268, September 1957.

\bibitem{Moura:2008}
L.~de~Moura and N.~Bj{\o}rner.
\newblock {Z3}: An efficient {SMT} solver.
\newblock In {\em TACAS}, pages 337--340. Springer-Verlag, 2008.

\bibitem{DBLP:journals/cacm/DershowitzM79}
N.~Dershowitz and Z.~Manna.
\newblock Proving termination with multiset orderings.
\newblock {\em Commun. ACM}, 22(8):465--476, 1979.

\bibitem{FLHT01}
C.~Ferm{\"u}ller, A.~Leitsch, U.~Hustadt, and T.~Tammet.
\newblock Resolution decision procedures.
\newblock In A.~Robinson and A.~Voronkov, editors, {\em Handbook of Automated
  Reasoning}, chapter~25, pages 1791--1850. Elsevier, 2001.

\bibitem{GrafSaidi97ConstructionAbstractStateGraphsPVS}
S.~Graf and H.~Saidi.
\newblock Construction of abstract state graphs with {PVS}.
\newblock In {\em CAV}, pages 72--83, 1997.

\bibitem{DBLP:conf/pldi/GrebenshchikovLPR12}
S.~Grebenshchikov, N.~P. Lopes, C.~Popeea, and A.~Rybalchenko.
\newblock Synthesizing software verifiers from proof rules.
\newblock In {\em PLDI}, 2012.

\bibitem{DBLP:conf/popl/GuptaPR11}
A.~Gupta, C.~Popeea, and A.~Rybalchenko.
\newblock Predicate abstraction and refinement for verifying multi-threaded
  programs.
\newblock In {\em POPL}, 2011.

\bibitem{DBLP:conf/aplas/GuptaPR11}
A.~Gupta, C.~Popeea, and A.~Rybalchenko.
\newblock Solving recursion-free {Horn} clauses over {LI+UIF}.
\newblock In {\em APLAS}, pages 188--203, 2011.

\bibitem{DBLP:conf/popl/HeizmannHP10}
M.~Heizmann, J.~Hoenicke, and A.~Podelski.
\newblock Nested interpolants.
\newblock In {\em POPL}, 2010.

\bibitem{Henzinger:2004}
T.~A. Henzinger, R.~Jhala, R.~Majumdar, and K.~L. McMillan.
\newblock Abstractions from proofs.
\newblock In {\em POPL}, pages 232--244. ACM, 2004.

\bibitem{DBLP:conf/cav/LalQL12}
A.~Lal, S.~Qadeer, and S.~K. Lahiri.
\newblock Corral: A solver for reachability modulo theories.
\newblock In {\em CAV}, 2012.

\bibitem{Lewis:1978:RSC:322047.322059}
H.~R. Lewis.
\newblock Renaming a set of clauses as a {Horn} set.
\newblock {\em J. ACM}, 25(1):134--135, Jan. 1978.

\bibitem{tree-interpolants}
K.~L. McMillan.
\newblock {iZ3} documentation.
\newblock
  \\\url{http://research.microsoft.com/en-us/um/redmond/projects/z3/iz3documen%
tation.html}.

\bibitem{DBLP:conf/cav/McMillan03}
K.~L. McMillan.
\newblock Interpolation and {SAT}-based model checking.
\newblock In {\em CAV}, 2003.

\bibitem{DBLP:conf/csl/McMillan04}
K.~L. McMillan.
\newblock Applications of craig interpolation to model checking.
\newblock In J.~Marcinkowski and A.~Tarlecki, editors, {\em CSL}, volume 3210
  of {\em Lecture Notes in Computer Science}, pages 22--23. Springer, 2004.

\bibitem{DBLP:conf/cav/McMillan06}
K.~L. McMillan.
\newblock Lazy abstraction with interpolants.
\newblock In {\em CAV}, 2006.

\bibitem{duality}
K.~L. McMillan and A.~Rybalchenko.
\newblock Solving constrained {Horn} clauses using interpolation.
\newblock Technical Report MSR-TR-2013-6, Jan. 2013.
\newblock \url{http://research.microsoft.com/apps/pubs/default.aspx?id=180055}.

\bibitem{lopstr07}
M.~M{\'e}ndez-Lojo, J.~A. Navas, and M.~V. Hermenegildo.
\newblock A flexible, (c)lp-based approach to the analysis of object-oriented
  programs.
\newblock In {\em LOPSTR}, pages 154--168, 2007.

\bibitem{disjInterpolantsTR}
P.~{R{\"u}mmer}, H.~{Hojjat}, and V.~{Kuncak}.
\newblock {Disjunctive Interpolants for {Horn}-Clause Verification (Extended
  Technical Report)}.
\newblock {\em ArXiv e-prints}, Jan. 2013.
\newblock \url{http://arxiv.org/abs/1301.4973}.

\bibitem{sfs2012}
O.~Sery, G.~Fedyukovich, and N.~Sharygina.
\newblock Interpolation-based function summaries in bounded model checking.
\newblock In {\em Haifa Verification Conference (HVC)}, Haifa, 2011. Springer.

\bibitem{SuterETAL11SatisfiabilityModuloRecursivePrograms}
P.~Suter, A.~S. K\"{o}ksal, and V.~Kuncak.
\newblock Satisfiability modulo recursive programs.
\newblock In {\em Static Analysis Symposium (SAS)}, 2011.

\bibitem{DBLP:journals/ase/TaghdiriJ07}
M.~Taghdiri and D.~Jackson.
\newblock Inferring specifications to detect errors in code.
\newblock {\em Autom. Softw. Eng.}, 14(1):87--121, 2007.

\end{thebibliography}
}

\appendix

\end{document}